\documentclass[lettersize,journal]{IEEEtran}
\usepackage{amsthm,amsmath,indentfirst,amssymb}
\usepackage{algorithmic}
\usepackage{array}
\usepackage[caption=false,font=normalsize,labelfont=sf,textfont=sf]{subfig}
\usepackage{textcomp}
\usepackage{stfloats}
\usepackage{url}
\usepackage{verbatim}
\usepackage{graphicx}

\newtheorem{lemma}{Lemma}
\newtheorem{theorem}[lemma]{Theorem}
\newtheorem{coro}[lemma]{Corollary}
\newtheorem{definition}[lemma]{Definition}

\newtheorem{rmk}[lemma]{Remark}
\newtheorem{property}[lemma]{Property}

\usepackage{tikz}
\usepackage{algorithm}
\usepackage{algorithmic}
\usepackage{color}
\def\BibTeX{{\rm B\kern-.05em{\sc i\kern-.025em b}\kern-.08em
    T\kern-.1667em\lower.7ex\hbox{E}\kern-.125emX}}
\usepackage{balance}
\begin{document}
\title{A Constant-Approximation Algorithm for Budgeted Sweep Coverage with Mobile Sensors}
\author{Wei~Liang,~\IEEEmembership{}
	Shaojie~Tang,~\IEEEmembership{Member,~IEEE,}
	and~Zhao~Zhang,~\IEEEmembership{}
	\IEEEcompsocitemizethanks{\IEEEcompsocthanksitem W. Liang is with School of Mathematical Sciences, Zhejiang Normal University, Jinhua, Zhejiang, 321004, People's Republic of China.\protect\\
		E-mail: lvecho1019@zjnu.edu.cn
		\IEEEcompsocthanksitem S. Tang is with Naveen Jindal School of Management, University of Texas at Dallas, Richardson, Texas 75080, USA.\protect\\
		E-mail: shaojie.tang@utdallas.edu
		\IEEEcompsocthanksitem Z. Zhang is the corresponding author, with School of Mathematical Sciences, Zhejiang Normal University, Jinhua, Zhejiang, 321004, People's Republic of China.
		This research is supported by National Natural Science Foundation of China (U20A2068).\protect\\
		E-mail: hxhzz@sina.com

	}
}

\markboth{Journal of \LaTeX\ Class Files,~Vol.~18, No.~9, September~2020}%
{How to Use the IEEEtran \LaTeX \ Templates}

\maketitle

\begin{abstract}
	In this paper, we present the first constant-approximation algorithm for {\em budgeted sweep coverage problem} (BSC). The BSC involves designing routes for a number of mobile sensors (a.k.a. robots) to periodically collect information as much as possible from points of interest (PoIs). To approach this problem, we propose to first examine the {\em multi-orienteering problem} (MOP). The MOP aims to find a set of $m$ vertex-disjoint paths that cover as many vertices as possible while adhering to a budget constraint $B$. We develop a constant-approximation algorithm for MOP and utilize it to achieve a constant-approximation for BSC. Our findings open new possibilities for optimizing mobile sensor deployments and related combinatorial optimization tasks.
\end{abstract}

\begin{IEEEkeywords}
	Sweep Cover, Approximation Algorithm, Budgeted Cover, Orienteering.
\end{IEEEkeywords}

\section{Introduction}

%
%
%
%
\IEEEPARstart{T}{he} (sensor) coverage problem has applications in various domains, including environmental monitoring, surveillance, intrusion detection, precision agriculture, and healthcare \cite{AakashTON,XuTON,GuoTMC,Shen-IoTJ}. While much existing research focuses on deploying static sensors to provide continuous coverage for points of interest (PoIs) \cite{ran2021,Huang-2021,zhang2022-area,Fan2014,YangTON}, this approach often requires substantial resources. In many real-world applications, such as police patrolling and wildlife conservation monitoring, periodic coverage of PoIs is sufficient. To achieve periodic coverage, we need to design trajectories for mobile sensors (a.k.a. robots) that guide them to visit PoIs periodically. A PoI $v$ is considered {\em sweep-covered} if it is visited at least once within every time period $t_v$, where $t_v$ is referred to as the {\em sweep-period} of $v$. Depending on the scenario, the optimization objective can be either sensor-oriented, aiming to minimize the number \cite{Li-2011,Liu-2021,WangTCS} or speed \cite{zhao-2012,GU-2006} of mobile sensors such that sufficient PoIs are sweep-covered, or target-oriented, aiming to minimize the maximum sweep-period \cite{Gao2018,Gao2022} or maximize the number of sweep-covered PoIs using limited resources \cite{Diyan2021,Liang-charge}.

The study of the sweep coverage problem was initially undertaken by \cite{Chen2008}. Their goal was to use the minimum number of mobile sensors to achieve sweep coverage for {\em all} PoIs (MinSSC). They proved that MinSSC cannot be approximated within a factor of 2 unless $P=NP$. Subsequently, \cite{Gorain2015} proposed a 3-approximation algorithm for MinSSC assuming uniform sweep-period and uniform speed.

In many real-world applications, it may not be possible to allocate enough sensors to meet the coverage requirements of each PoI. Thus, the research focus shifted towards maximizing coverage efficiency with limited resources, leading to the investigation of the {\em budgeted sweep coverage problem} (BSC). In BSC, the objective is to sweep-cover as many PoIs as possible with a limited budget of mobile sensors. Furthermore, when every PoI is assigned a weight, the objective is to maximize the total weight of sweep-covered PoIs, which is referred to as the {\em weighted budgeted sweep coverage} (WBSC) problem.



Reference \cite{Huang2018} marked a pioneering effort in the study of WBSC. Assuming that each mobile sensor can choose only one route from a limited number of alternative routes, the authors provided a $(1-\frac 1e)$-approximation algorithm. Note that when presented with a polynomial number of alternate routes for each mobile sensor, the problem is reduced to a maximum set cover problem. However, in general, the number of routes can be exponential, and multiple mobile sensors may collaborate on a single route. \cite{Diyan2021} explored the WBSC problem on a line (WBSC-L), where all PoIs are located on a line. They introduced three approximation algorithms tailored to various mobile sensor speed scenarios. In the case where all mobile sensors operate at the same speed (a common assumption in many existing studies), \cite{Liang2023} developed a \( (4,1/2) \)-bicriteria algorithm for BSC. I.e., their algorithm achieves an approximation ratio of \( 1/2 \) but compromises feasibility by a factor of 4. This raises the question we aim to tackle in this paper: Is there a feasible constant-factor approximation algorithm for the general BSC?



\subsection{Related Work}
Starting from \cite{Chen2008}, the sweep coverage problem has been extensively investigated. In the seminal paper \cite{Chen2008}, it was assumed that the optimal sweep-route corresponds to a shortest Hamiltonian cycle. Based on such an assumption, they reduced the problem to the {\em traveling salesman problem} (TSP), proved that achieving a better than 2-approximation for MinSSC is unlikely unless $P=NP$, and designed a 2-approximation algorithm for MinSSC.
Later, \cite{Gorain2015} discovered that grouping the PoIs can significantly reduce the number of required sensors. In their algorithm, PoIs are divided into groups and each group of PoIs are jointly sweep-covered by a same set of sensors, which prevents wastage of sensors caused by traveling between distant PoIs. They proposed a 3-approximation algorithm using a {\em minimum spanning forest} to group the PoIs efficiently.

The budgeted version of sweep coverage (BSC) was proposed by \cite{Huang2018}. They proposed a $(1-\frac{1}{e})$-approximation algorithm for WBSC, assuming each mobile sensor is associated with a finite set of alternative routes. In \cite{Nie-budget}, two heuristic algorithms were presented for WBSC. \cite{Diyan2021} studied the WBSC problem on a line (WBSC-L), considering three mobile sensor speed conditions: ($\romannumeral1$) uniform speed, ($\romannumeral2$) constant number of different speeds, and ($\romannumeral3$) moderate number of different speeds. They designed an exact algorithm using dynamic programming for case ($\romannumeral1$), which was extended to a $\frac{1}{2}$-approximation algorithm for case ($\romannumeral2$), and a $(\frac{1}{2}-\frac{1}{2e})$-approximation algorithm for case ($\romannumeral3$) using linear programming and randomization.

\cite{Liang2021} introduced the {\em prize-collecting minimum sensor sweep coverage problem} (PCMinSSC). As BSC, the PCMinSSC does not require all PoIs to be sweep-covered. At the same time penalties should be paid for uncovered PoIs. The goal of PCMinSSC is to minimize the cost of sensors along with the penalties incurred by uncovered PoIs. The authors devised a 5-approximation algorithm for a specific version called {\em prize-collecting minimum sensor sweep coverage with base stations} (PCMinSSC$_{BS}$), which assumes that each mobile sensor has to be linked to a base station and the number of base stations is a constant. Recently \cite{Liang2023a} relaxed the base station assumption and presented a 5-approximation algorithm for PCMinSSC. Their algorithm is based on a 2-LMP algorithm for a new combinatorial optimization problem called {\em prize-collecting forest with $k$ components problem} (PCF$_k$). PCF$_k$ aims to find a forest of exactly $k$-components such that the cost of the forest plus the penalties on those vertices not spanned by the forest is as small as possible.

Our approach to BSC relies on the solution of a new combinatorial optimization problem called the {\em multi-orienteering problem}  (MOP). It is a generalization of the {\em un-rooted orienteering problem} (UOP). Given a metric graph and a budget $B$, the goal of UOP is to find a simple path containing the maximum number of vertices and costing no more than $B$. The UOP can be solved using an algorithm for the {\em rooted orienteering problem} (ROP), where a starting vertex $s$ is pre-given and the path should begin at $s$.  \cite{blum2007} presented a $4$-approximation algorithm for ROP.
When points are in a fixed-dimensional Euclidean space, a 2-approximation algorithm was proposed in \cite{Arkin-1998}, which was later improved to a PTAS by  \cite{Chen-o-2008}. \cite{Bansal-2004} considered a more general orienteering problem, the {\em point-to-point orienteering problem} (P2P-OP), where the path is required to end at a specified vertex $t$, in addition to starting at a given vertex $s$. The authors designed a $3$-approximation algorithm for P2P-OP, which was later improved to $2+\varepsilon$ by \cite{Chekuri2012}.
\cite{Xu-MO} addressed a {\em generalized team orienteering problem}, aiming to find multiple $s$-$t$ paths, each costing at most $B$, such that the total number of spanned vertices is maximized. By utilizing the $(2+\varepsilon)$-approximation algorithm for P2P-OP, they obtained a $(1-(1/e)^{\frac{1}{2+\varepsilon}})$-approximation.
For more variants of the orienteering problem, refer to the surveys \cite{Vansteenwegen2011,Gunawan2016}.

Our result for the MOP is based on the {\em minimum weight vertex-disjoint $m$-paths spanning at least $k$ vertices problem} ($k$-MinWP$_m$). The $k$-MinWP$_m$ asks for $m$ vertex-disjoint paths spanning at least $k$ vertices such that their total length is as small as possible. When $m=1$, the $k$-MinWP$_1$ can be viewed as an {\em un-rooted $k$-stroll problem} ($k$-SP). The {\em point-to-point version of the $k$-stroll problem} (P2P-$k$-SP) is to find a minimum length $s$-$t$ path that visits at least $k$ vertices in a given metric graph $G$, where $s$ and $t$ are pre-specified vertices. There are many studies on the P2P-$k$-SP \cite{bateni2013,Chaudhuripath,Chekuri2012,Nagarajan2007}. Our $k$-MinWP$_m$ generalizes the $k$-SP by requiring more paths (the number of paths $m$ might not be a constant).


\subsection{Our Contribution}
In this paper, we address the open problem of whether a constant-approximation exists for the BSC \cite{Liang2023}, achieving an approximation ratio of $\left(0.0116-O(\varepsilon)\right)$. 

Unlike past research on various sweep coverage challenges that typically employ vertex-disjoint {\em trees} for grouping PoIs, our strategy utilizes vertex-disjoint {\em paths} rather than trees. By allowing sensors to operate independently along these paths (instead of operating collaboratively on cycles as in the past researches), we attain the first constant approximation without compromising feasibility.

In order to solve BSC, we propose the MOP and present a $(0.035-O(\varepsilon))$-approximation for MOP. The key step is a bicriteria result for the $k$-MinWP$_m$, a new problem proposed in this paper. Inspired by \cite{Blum1996}, we address $k$-MinWP$_m$ using its associated ``prize-collecting'' problem, the {\em prize-collecting vertex-disjoint multi-paths problem} (PCP$_m$).


In our work on MOP, the main challenge lies in managing a given budget as the upper bound shared across all paths, with no specified starting and ending vertices. When the number of paths is not a constant, the combinations of starting and ending vertices increase exponentially. Additionally, an overall budget constraint cannot be easily decomposed into individual budget constraints, as it is unclear how much budget should be allocated to each path, and guessing individual budgets may take exponential time. A similar ongoing challenge exists for $k$-MinWP$_m$, where the complexity also arises from unspecified starting and ending vertices.

Our paper builds the algorithm step by step (see Figure \ref{fig:1}), overcoming challenges posed by non-constant number of paths and the absence of pre-given roots. All aforementioned results might be of independent interest and beneficial for network design and multiple vehicles routing problems. Overall, our results provide significant advancements in solving the BSC and open up new avenues of research in related problems.

\begin{figure} [h]
\centering
\tikzstyle{format}=[rectangle,draw,thin,fill=white]
\tikzstyle{test}=[diamond,aspect=2,draw,thin]
\tikzstyle{point}=[coordinate,on grid,]
\begin{tikzpicture}
	\node[format] (start){2-LMP for PCF$_m$};
	\node[format,below of=start,node distance=12mm] (define){4-LMP for PCP$_m$};
	\node[format,below of=define,node distance=12mm] (PCFinit){Bicriteria algorithm for $k$-MinWP$_m$};
	\node[format,below of=PCFinit,node distance=12mm] (DS18init){$(0.035-O(\epsilon))$-approximation for MOP};
	\node[format,below of=DS18init,node distance=12mm] (LCDinit){$(0.0116-O(\epsilon))$-approximation for BSC};
	
	\draw[->] (start)--(define);
	\draw[->] (define)--(PCFinit);
	\draw[->](PCFinit)--(DS18init);
	\draw[->](DS18init)--(LCDinit);
\end{tikzpicture}
\caption{The outline of finding a solution to BSC.}
\label{fig:1}
\end{figure}
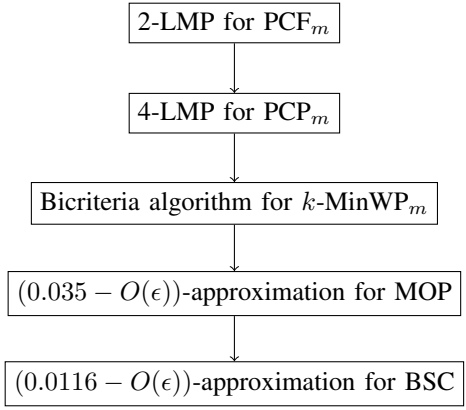

\subsection{Organization}

The rest of this paper is organized as follows: The problem is formally defined in Section 2, together with some preliminary results. The constant-approximation for MOP is presented in Section 3, based on which the approximation algorithm for BSC is given in Section 4. Section 5 concludes the paper.

\section{Problem formulation and preliminaries}
This section formally defines the budgeted sweep coverage problem and the multi-orienteering problem, together with some related terminologies.

\begin{definition}[sweep coverage]
{\em Let $G=(V,E,w)$ be a graph on vertex set $V$, edge set $E$ and edge weight function $w:E\mapsto \mathbb R^+$ which is a metric on $G$. Assume that mobile sensors can move along the edges of $G$ at the same speed $a$. For a positive real number $t$ called {\em sweep period}, a vertex $v$ is said to be {\em sweep-covered} if $v$ is visited by the mobile sensors at least once in every time period $t$.}
\end{definition}


\begin{definition}[budgeted sweep coverage (BSC)]\label{pro524-01}
{\rm Given a metric graph $G=(V,E,w)$ and a positive integer $N$, the goal of BSC is to design trajectories for $N$ mobile sensors such that the number of sweep-covered vertices is maximized.}
\end{definition}

For a real number $0<\gamma \leq 1$, a polynomial time algorithm for a maximization problem is considered a {\em $\gamma$-approximation algorithm} if for any instance $I$ of the problem, the output of the algorithm has value at least $\gamma\cdot opt(I)$, where $opt(I)$ is the optimal value of the instance. While for a minimization problem and a real number $\beta\geq 1$, a polynomial time algorithm is said to be a {\em $\beta$-approximation algorithm} if the output solution has value at most $\beta\cdot opt(I)$.

As a step stone for solving BSC, we propose the multi-orienteering problem as follows.

\begin{definition}[multi-orienteering problem (MOP)]\label{pro624-1}
{\rm Given a metric graph $G=(V,E,w)$, a positive integers $m$ and a positive real number $B$, the goal of MOP is to find a set $\mathcal P_m$ of $m$ vertex-disjoint paths in $G$ spanning the maximum number of vertices such that $w(\mathcal P_m)=\sum_{P\in \mathcal P_m}w(P)\leq B$, where $w(P)=\sum_{e\in E(P)}w(e)$.}
\end{definition}

We solve MOP by considering the following problem.

\begin{definition}[minimum weight vertex-disjoint multi-paths spanning at least $k$ vertices ($k$-MinWP$_m$)]\label{pro55-1}
{\rm Given a metric graph $G=(V,E,w)$ and two integers $m,k$ with $m\leq k\leq |V|$, the goal of $k$-MinWP$_m$ is to find a set $\mathcal P_m$ of $m$ vertex-disjoint paths spanning at least $k$ vertices such that the weight $w(\mathcal P_m)=\sum_{P\in \mathcal P_m}w(P)$ is minimized.}
\end{definition}

For $k$-MinWP$_m$, we propose a bicriteria approximation algorithm that allows for limited violations of constraints while ensuring a provable approximation guarantee. For real numbers $\beta\geq 1$ and $0<\alpha\leq 1$, an algorithm for $k$-MinWP$_m$ is said to be an {\em $(\alpha,\beta)$-bicriteria approximation} if for any $k$-MinWP$_m$ instance $I$, it can always compute in polynomial time a solution $\mathcal P_m$ satisfying
$$| V(\mathcal P_m)|\geq \alpha \cdot k \ \mbox{and}\ w(\mathcal P_m)\leq \beta\cdot opt(I),$$
where $opt(I)$ is the optimal value of instance $I$.


The bicriteria algorithm for $k$-MinWP$_m$ is based on an approximation algorithm for a prize-collecting version of the vertex-disjoint multi-path problem, which is defined as follows.

\begin{definition}[prize-collecting vertex-disjoint multi-paths (PCP$_m$)]\label{pro526-01}
{\rm Given a metric graph $G=(V,E,w)$ and a penalty function on vertex set $\pi: V\mapsto \mathbb{R}^+$, the goal of PCP$_m$ is to find a set $\mathcal P_m$ of $m$ vertex-disjoint paths such that the weight of $\mathcal P_m$ plus the penalty of vertices not spanned by $\mathcal P_m$ is minimized, that is, $\min\{w(\mathcal P_m)+\pi(V\setminus V(\mathcal P_m))\}$, where $\pi(V\setminus V(\mathcal P_m))=\sum_{v\notin V(\mathcal P_m)}\pi(v)$. An instance of PCP$_m$ is written as $(G,w,\pi)$.}
\end{definition}

The approximation algorithm for PCP$_m$ is built upon an approximation algorithm for a prize-collecting forest problem, which is defined as follows.

\begin{definition}[prize-collecting forest with $m$ components (PCF$_m$)]\label{pro113-01}
{\rm PCF$_m$ is similar to PCP$_m$ except that the goal is to find a forest containing exactly $m$ components, that is, $m$ vertex disjoint trees instead of $m$ vertex-disjoint paths.}
\end{definition}

We obtain  $r$-LMP algorithms for both PCP$_m$ and PCF$_m$, where the definition of $r$-LMP is given as follows.

\begin{definition}[Lagrangian multiplier preserving algorithm with factor $r$ ($r$-LMP)]
{\rm For a real number $r\geq 1$, an algorithm for PCP$_m$ (resp. PCF$_m$)  is said to be an $r$-LMP if for any instance $I$ of PCP$_m$ (resp. PCF$_m$), the algorithm can output a set $\mathcal{T}$ of $m$ paths (resp. trees)  in polynomial time such that
	$$w(\mathcal T)+r\cdot\pi(V\setminus V(\mathcal{T}))\leq r\cdot opt(I).$$}
\end{definition}


In this paper, we utilize a well-known short-cutting strategy to derive a path $P$ from a tree. This method, which is used to obtain a 2-approximate solution for the traveling salesman problem \cite{Vazirani}, involves doubling every edge in the tree to form an Eulerian graph. By traversing a closed walk that includes every edge exactly once and shortcutting repeated vertices, we obtain a cycle. Removing any edge from this cycle results in a path $P$ that covers each vertex of $T$ exactly once. When dealing with a metric underlying graph, the triangle inequality ensures that $w(P) \leq 2w(T)$.

\section{A Constant Approximation Algorithm for MOP}\label{sec-003}
In this section, we first give a bicriteria algorithm for $k$-MinWP$_m$, based on which a constant-approximation algorithm is obtained for MOP.

\subsection{A bicriteria algorithm for $k$-MinWP$_m$}\label{subsec-001}

The algorithm for $k$-MinWP$_m$ calls a 4-LMP algorithm for the PCP$_m$ problem as
a subroutine, which makes use of a 2-LMP algorithm for PCF$_m$. A detailed design of this 4-LMP algorithm is shown in the proof of the following theorem.

\begin{theorem}\label{theo610-1}
{\rm PCP$_m$ admits a $4$-LMP (denoted as $\mathcal A_{PC}$) and $\mathcal A_{PC}$ executes in time $O(|V|^4)$.}
\end{theorem}
\begin{proof}
We first use the 2-LMP algorithm for PCF$_m$ in \cite{Liang2023a} to compute a forest $F_m$. For each component $T_i$ of $F_m$, we use the short-cutting strategy to obtain a path $P_i$ with $w(P_i)\leq 2w(T_i)$.
Denote by $opt_{F}$ and $opt_{P}$ the optimal values of the PCP$_m$ instance and the PCF$_m$ instance, respectively. Then, the collection of paths $\mathcal P_m=\{P_1,\ldots,P_m\}$ satisfies
\begin{align*}
	w(\mathcal P_m)+4\pi(V\setminus\mathcal{P}_m)
	&\leq 2w(F_m)+4\pi(V\setminus F_m)\\
	&\leq 4opt_{F}\leq 4opt_{P},
\end{align*}
where the second inequality holds because $F_m$ is a 2-LMP solution for the PCF$_m$ instance, and the third inequality holds because any solution to the PCP$_m$ instance is a feasible solution to the PCF$_m$ instance.
\end{proof}

In the following, $0<\alpha<1$ is an adjustable parameter. 

\subsubsection{Constructing $m$-vertex disjoint paths $\mathcal P_m$}

The construction is based on the following lemma.

\begin{lemma}\label{lemma-325-1}
Suppose a $k$-MinWP$_m$ instance has a feasible solution with cost at most $L$. Use the 4-LMP algorithm $\mathcal A_{PC}$ on the instance $(G,w,L/(1-\alpha)k)$ of PCP$_m$, one obtains a set $\mathcal P_{m,L}$ of $m$ vertex-disjoint paths. Then $\mathcal P_{m,L}$ spans more than $\alpha k$ vertices and the cost of $\mathcal P_{m,L}$ is at most $\frac{4(p_{m,L}-\alpha k)}{(1-\alpha)k}L$, where $p_{m,L}$ is the number of vertices spanned by the path set $\mathcal P_{m,L}$.
\end{lemma}

\begin{proof}
	Since $\mathcal A_{PC}$ is a 4-LMP algorithm for PCP$_m$, the set of paths it returns satisfies
	\begin{align}\label{eq3-24-1}
		w(\mathcal P_{m,L})+(n-p_{m,L})\frac{4L}{(1-\alpha)k}\leq 4opt_{PCP_m,L},
	\end{align}
	where $opt_{PCP_m,L}$ is the optimal value of instance $(G,w,L/(1-\alpha)k,m)$.
	
	Let $\mathcal P'_{m,L}$ be a feasible solution to the $k$-MinWP$_m$ instance with cost at most $L$, whose existence is guaranteed by the condition of the lemma. Suppose $\mathcal P'_{m,L}$ spans $p'_{m,L}$ vertices. Viewing $\mathcal P'_{m,L}$ as a feasible solution to the PCP$_m$ instance $(G,w,L/(1-\alpha)k)$, we have
	$$
	opt_{PCP_m,L}\leq L+(n-p'_{m,L})\frac{L}{(1-\alpha)k}\leq L+(n-k)\frac{L}{(1-\alpha)k}.
	$$
	Combining this with inequality \eqref{eq3-24-1}, we have
	\begin{align}
		w(\mathcal P_{m,L})+(n-p_{m,L})\frac{4L}{(1-\alpha)k}\leq 4L+(n-k)\frac{4L}{(1-\alpha)k}\nonumber.
	\end{align}
	Rearranging terms,
	\begin{align}\label{eq3-24-3}
		w(\mathcal P_{m,L})\leq \frac{4(p_{m,L}-\alpha k)}{(1-\alpha)k}L.
	\end{align}
	
	As a byproduct of \eqref{eq3-24-3}, since $w(\mathcal P_{m,L})\geq 0$, we have $p_{m,L}>\alpha k$, that is, $\mathcal P_{m,L}$ spans at least $\alpha k$ vertices.
\end{proof}

In particular, Lemma \ref{lemma-325-1} holds for $L=opt_{m,k}$, where $opt_{m,k}$ is the optimal cost of the $k$-MinWP$_m$.
Although we do not know $opt_{m,k}$, a binary search method (the pseudo code is given in Algorithm \ref{algo4}, whose ideas are described in the proof of Lemma \ref{coro0108-1}) can give it an estimated bound.

\begin{algorithm}
	\caption {Binary search to approximate $opt_{m,k}$.}
	\hspace*{0.02in}\raggedright{\bf Input:} $G=(V,w,\pi)$, two positive integers $k\geq m$ and two real numbers $\varepsilon, \alpha\in(0,1)$;
	
	\hspace*{0.02in}\raggedright{\bf Output:} Real number $L_1$ that approximates $opt_{m,k}$ and the corresponding path-set $\mathcal P_m=\mathcal P_{m,L_1}$.
	
	\begin{algorithmic}[1]
		\STATE $Q$ $\leftarrow$ the total length of the longest $k-m$ edges of $G$;
		\STATE $L_1\leftarrow Q$, $L_2\leftarrow 0$;
		\WHILE {$L_1-L_2> \varepsilon$}
		\STATE $L\leftarrow$ $\frac{L_1+L_2}{2}$
		\STATE $p_{m,L}\leftarrow$ the number of vertices spanned by the $m$-vertex disjoint paths $\mathcal P_{m,L}$ computed by $\mathcal{A}_{PC}$ on PCP$_m$ instance $(G,w,L/(1-\alpha)k)$;
		\IF{$p_{m,L}\leq\alpha k$}
		\STATE $L_2\leftarrow L$
		\ELSE
		\STATE $L_1\leftarrow L$
		\ENDIF
		\ENDWHILE
		\RETURN $L_1$ and $\mathcal P_m\leftarrow \mathcal P_{m,L_1}$.
	\end{algorithmic}\label{algo4}
\end{algorithm}

The number $L_1$ output by Algorithm \ref{algo4} satisfies $L_1\leq opt_{m,k}+\varepsilon$ and the path-set $\mathcal P_{m,L_1}$ satisfies the performance bounds listed in the following lemma.

\begin{lemma}\label{coro0108-1}
Given a $k$-MinWP$_m$ instance $(G,m,k)$, Algorithm \ref{algo4} computes $m$ vertex-disjoint paths $\mathcal P_m$ spanning more than $\alpha k$ vertices with cost at most $\frac{4(p_m-\alpha k)}{(1-\alpha)k}(opt_{m,k}+\varepsilon),$ where $p_m$ is the number of vertices spanned by $\mathcal P_m$.
\end{lemma}

\begin{proof}
	Let $Q$ be the total length of the $k-m$ longest edges in graph $G$. Then $opt_{m,k}\in [0,Q]$ (because $m$ vertex-disjoint paths spanning $k$ vertices contains exactly $k-m$ edges). The algorithm conducts a binary search on $L\in [0,Q]$, using the 4-LMP algorithm $\mathcal A_{PC}$ on the instance $(G,w,L/(1-\alpha)k)$ of PCP$_m$, find out two numbers $L_1,L_2$ such that
	\begin{equation}\label{eq0108-2}
		p_{m,L_1}> \alpha k,\ p_{m,L_2}\leq\alpha k \ \mbox{and}\ 0<L_1-L_2\leq \varepsilon.
	\end{equation}
Note that the larger the penalty is, the more vertices are spanned by the computed path, so $L_1>L_2$.
	
	Let $\mathcal P_m=\mathcal P_{m,L_1}$. Then $\mathcal P_m$ spans more than $\alpha k$ vertices. By Lemma \ref{lemma-325-1}, taking $L\geq opt_{m,k}$ can always yield a solution spanning more than $\alpha k$ vertices. Since $p_{m,L_2}\leq\alpha k$, we have $L_2<opt_{m,k}$. As a consequence, $L_1<opt_{m,k}+\varepsilon$, and thus $w(\mathcal P_m)=w(\mathcal P_{m,L_1})\leq \frac{4(p_{m,L_1}-\alpha k)}{(1-\alpha)k}L_1\leq \frac{4(p_m-\alpha k)}{(1-\alpha)k}(opt_{m,k}+\varepsilon)$. The lemma is proved.
\end{proof}

\begin{rmk}\label{rmk0812-1}
Let $T_{PC}$ be the running time of $\mathcal A_{PC}$. Finding the desired $L_1$ takes time $O( T_{PC}\log \frac{Q}{\varepsilon})$. Note that in a sweep coverage problem, w.o.l.g., we can assume that $w(e)\leq 1$ by both scaling the speed and the edge weight, without changing the approximation ratio. In this way, the upper bound of $opt_{m,k}$ is $k-m$ and Algorithm \ref{algo4} takes time  $O(T_{PC}\log\frac{k-m}{\varepsilon})$.
\end{rmk}

Note that the cost of the $m$-vertex disjoint paths $\mathcal P_m$, whose existence is guaranteed by Lemma \ref{coro0108-1}, is dependent on $p_m$, the number of vertices spanned by $\mathcal P_m$. If $p_m\leq 2k$, then by Lemma \ref{coro0108-1}, $\mathcal P_m$ spans more than $\alpha k$ vertices and $w(\mathcal P_m)\leq \frac{4(2-\alpha)}{1-\alpha}(opt_{m,k}+\varepsilon)$, implying that $\mathcal P_m$ is an $(\alpha, \frac{8-4\alpha}{1-\alpha}+O(\varepsilon))$-bicriteria approximation to the $k$-MinWP$_m$ instance. 

However, when $p_m$ is large, the approximation ratio might be large. To mitigate this, we proceed to {\em trim} the paths in the subsequent subsection. This trimming ensures that the paths span at least $k$ vertices, attaining an approximation ratio of at most $\frac{16}{(1-\alpha)}+\varepsilon$. The details are presented in Algorithm \ref{algo1} and the ideas are explained as follows.

\subsubsection{Dealing with the case when $\mathcal P_m$ spans many vertices}\label{sec0722-1}

We iteratively divide each path into two sub-paths that span almost equal number of vertices, and discard the one with the larger {\em cost-to-vertex} ratio. The process terminates when the number of remaining vertices is less than $2k$. Related notations are given as follows.

A path is {\em trivial} if it contains only one vertex. Consider a non-trivial path $P$ containing $q$ vertices. For even $q$, denote the sub-path induced by the leftmost (resp. rightmost) $q/2$ vertices as $P^{(0)}$ (resp. $P^{(1)}$). For $j\in\{0,1\}$, denote $q^{(j)}$ the number of vertices contained in $P^{(j)}$, and let $w^{(j)}$ be the weight of $P^{(j)}$ plus the weight of the {\em middle edge} between $P^{(0)}$ and $P^{(1)}$. If $q$ is odd, then removing the {\em middle vertex} $v$ results in two sub-paths $P^{(0)}$ and $P^{(1)}$, containing $q^{(0)}=q^{(1)}=(q-1)/2$ vertices, respectively. In this case, for $j\in \{0,1\}$, let $w^{(j)}$ be the weight of $P^{(j)}$ plus the weight of the edge connecting $P^{(j)}$ and $v$. These notations are visually illustrated in Figure \ref{fig330-1} below.


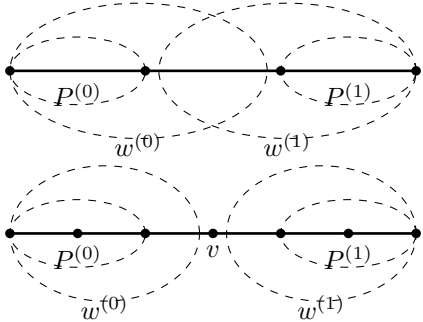
\begin{figure}[htpb]
\begin{center}
	\begin{tikzpicture}[scale=0.9]
		\draw [line width=1pt](-8,0)--(-2,0);
		\fill (-6,0) circle (2pt);
		\fill (-4,0) circle (2pt);
		\fill (-8,0) circle (2pt);
		\fill (-2,0) circle (2pt);
		\draw(-7,0)[dashed] ellipse(1 and 0.5)node[below = 1pt]{$P^{(0)}$};
		\draw(-6.1,0)[dashed] ellipse(1.9 and 1)node[below = 20pt]{$w^{(0)}$};
		\draw(-3,0)[dashed] ellipse(1 and 0.5)node[below = 1pt]{$P^{(1)}$};
		\draw(-3.9,0)[dashed] ellipse(1.9 and 1)node[below = 20pt]{$w^{(1)}$};
	\end{tikzpicture}
	\hskip 0.2cm
	\begin{tikzpicture}[scale=0.9]
		\draw [line width=1pt](0,0)--(6,0);
		\fill (0,0) circle (2pt);
		\fill (1,0) circle (2pt);
		\fill (2,0) circle (2pt);
		\fill (3,0) circle (2pt)node[below = 1pt]{$v$};
		\fill (4,0) circle (2pt);
		\fill (5,0) circle (2pt);
		\fill (6,0) circle (2pt);
		\draw(1,0)[dashed] ellipse(1 and 0.5)node[below = 1pt]{$P^{(0)}$};
		\draw(5,0)[dashed] ellipse(1 and 0.5)node[below = 1pt]{$P^{(1)}$};
		\draw(1.4,0)[dashed] ellipse(1.4 and 1)node[below = 20pt]{$w^{(0)}$};
		\draw(4.6,0)[dashed] ellipse(1.4 and 1)node[below = 20pt]{$w^{(1)}$};
	\end{tikzpicture}
	\vskip 0.2cm\caption {An illustration of $P^{(j)}$ and $w^{(j)}$ for $j\in\{0,1\}$.}\label{fig330-1}
\end{center}
\end{figure}

\begin{algorithm}[htpb]
\caption {Trimming paths in $\mathcal P_m$}
\hspace*{0.02in}\raggedright{\bf Input:} A set $\mathcal P_m$ of $m$ vertex-disjoint paths spanning more than $2k$ vertices.

\hspace*{0.02in}\raggedright{\bf Output:} A set $\mathcal P_m^{M}$ of $m$ vertex-disjoint paths spanning  $p_m^M\in[k,2k]$ vertices.

\begin{algorithmic}[1]
	\STATE $\mathcal P_m^M\leftarrow \mathcal P_m$;
	\WHILE {$\mathcal P_m^M$ spans more than $2k$ vertices}
	\FOR{each non-trivial path $P\in\mathcal P_m^M$}
	\STATE $j^*\leftarrow \arg\max_{j\in\{0,1\}}\{w^{(j)}/q^{(j)}\}$;\label{linej}
	\STATE $P\leftarrow P-P^{(j^*)}$ and update $\mathcal P_m^M$;
	\ENDFOR
	\ENDWHILE
	\RETURN $\mathcal P_m^{M}$.
\end{algorithmic}\label{algo1}
\end{algorithm}

The main result of this subsection is presented in Lemma \ref{lem0108-6}, which demonstrates that the trimmed path set $\mathcal P^M_m$ comprises a minimum of $k$ vertices while experiencing a substantial reduction in its weight. The proof of this lemma relies on the utilization of the following two technical lemmas (Lemma \ref{lemma-427-1} and Lemma \ref{lemma-427-2}). We first introduce some necessary notations. A path $P\in\mathcal P_m$ is said to be {\em short} if it is trimmed into a trivial path before the end of the algorithm, otherwise, it is called {\em long}. Denote by $\mathcal P^{0}$ and $\mathcal P^{>0}$ the sets of short paths and long paths, respectively, and let $\mathcal P^{\geq 0}=\mathcal P^{>0}\cup\mathcal P^{0}$. For a path $P\in\mathcal P_m$, denote by $P_l$ the remaining segment of $P$ at the end of the $l$-th iteration. In particular, $P_0=P$. Furthermore, denote by $P_M$ the remaining segment of $P$ at the end of the algorithm. Let $q(P_l)$ be the number of vertices spanned by $P_l$ and let $w(P_l)$ be the weight of $P_l$.

The proofs of Lemma \ref{lemma-427-1} and Lemma \ref{lemma-427-2} make use of the following observation.

\begin{property}\label{obser331-1}
	For a set of positive numbers $a_1,\ldots,a_t$ and $b_1,\ldots,b_t$, $$\min_{i=1,\ldots,t}\frac{a_i}{b_i}\leq\frac{a_1+a_2+\ldots+a_t}{b_1+b_2+\ldots+b_t}\leq \max_{i=1,\ldots,t}\frac{a_i}{b_i}.$$
\end{property}

\begin{lemma}\label{lemma-427-1}
For the set $\mathcal P^{>0}$ of long paths,
$$
\sum\limits_{P\in \mathcal P^{>0}}\frac{q(P_M)}{q(P_0)}w(P)\leq 2\cdot\frac{\sum_{P\in \mathcal P^{>0}} q(P_M)}{|V(\mathcal P^{>0})|}w(\mathcal P^{>0})
.$$
\end{lemma}

\begin{proof}
	For simplicity of notations, for each path $P_i\in\mathcal P$, denote by $q_{i,0}$ and $q_{i,M}$ the number of vertices spanned by $P_i$ and the remaining segment of $P_i$ at the end of Algorithm \ref{algo1}, respectively. Call $q_{i,M}/q_{i,0}$ as the {\em shrinking ratio} of path $P_i$. If we can show
	\begin{equation}\label{cla-428-1}
		\frac{q_{i,M}/q_{i,0}}{q_{j,M}/q_{j,0}}\leq 2\ \mbox{holds for any}\ P_i,P_j\in\mathcal P^{>0},
	\end{equation}
	then the lemma can be proved as follows. Suppose $P_{i_0}$ is the path of $\mathcal P^{>0}$ with the minimum shrinking ratio. Note that
	$$
	\frac{\sum_{P\in \mathcal P^{>0}}{q(P_M)}}{|V(\mathcal P^{>0})|}=\frac{\sum_{P_i\in \mathcal P^{>0}}{ q_{i,M}}}{\sum_{P_i\in \mathcal P^{>0}}q_{i,0}}
	$$
	can be viewed as an average shrinking ratio. So
	$$
	\frac{\sum_{P\in \mathcal P^{>0}}q(P_M)}{|V(\mathcal P^{>0})|}\geq \frac{q_{i_0,M}}{q_{i_0,0}}.
	$$
	Hence, if \eqref{cla-428-1} is true, then for any $i$, we have
	$$
	\frac{q_{i,M}}{q_{i,0}}\leq 2\cdot \frac{q_{i_0,M}}{q_{i_0,0}}\leq 2\cdot \frac{\sum_{P\in \mathcal P^{>0}}{q(P_M)}}{|V(\mathcal P^{>0})|},
	$$
	and thus
	\begin{align*}
		\sum\limits_{P\in \mathcal P^{>0}}{\frac{q(P_M)}{q(P_0)}}w(P)  &\leq
		2\cdot \frac{\sum_{P\in \mathcal P^{>0}}{ q(P_M)}}{|V(\mathcal P^{>0})|}\sum_{P\in\mathcal P^{>0}}w(P)\\
		&=2\cdot \frac{\sum_{P\in \mathcal P^{>0}}{ q(P_M)}}{|V(\mathcal P^{>0})|}w(\mathcal P^{>0}).
	\end{align*}
	
	The following part is devoted to the proof of property \eqref{cla-428-1}. Suppose every long path is trimmed $M$ times. It can be checked that the minimum shrinking ratio is $\frac1{2^M}$ (when $q_{i,0}$ is divisible by $2^M$) and the maximum shrinking ratio is $\frac{1}{2^M-\frac{2^{M}-1}{q_{i,M}}}$ (when $q_{i,M}\geq 2$ and every node with odd label $a$ has its parent labeled as $2a-1$). So, for any pair of long paths $P_i$ and $P_j$, we have
	$$
	\frac{q_{i,M}/q_{i,0}}{q_{j,M}/q_{j,0}}\leq \frac{1}{1-\frac{1-1/{2^M}}{q_{j,M}}}\leq \frac{1}{1-{\frac{1}{2}}}=2.
	$$
	The lemma is proved.
\end{proof}

\begin{lemma}\label{lemma-427-2}
For the path sets $\mathcal P^{>0}$ and $\mathcal P^{0}$,
$$
\frac{\sum_{P\in \mathcal P^{>0}}q(P_M)}{|V(\mathcal P^{>0})|}\leq \frac{\sum_{P\in \mathcal P^{>0}}q(P_M)+\sum_{P\in \mathcal P^{0}}q(P_M)}{|V(\mathcal P^{>0})|+|V(\mathcal P^{0})|}.
$$
\end{lemma}

\begin{proof}
	To prove the lemma, it suffices to prove that $\forall P_i\in\mathcal P^{> 0}$ and $P_j\in\mathcal P^{0}$,
	\begin{equation}\label{cla-52-1}
		\frac{q_{i,M_i}}{q_{i,0}}\leq \frac{q_{j,M_j}}{q_{j,0}},
	\end{equation}
	where $q_{i,M_i}$ (resp. $q_{j,M_j}$) denotes the number of vertices spanned by long (resp. short) path $P_i$ (resp. $P_j$) at the end of Algorithm~\ref{algo1}.
	If \eqref{cla-52-1} is true, then by Observation \ref{obser331-1}, the lemma can be proved.
	
	Note that $M_i=M$ and $q_{j,M_j}=1$. Furthermore, any short path $P_j$ is trimmed $M_j\leq M-1$ times. By the argument in the proof of the previous lemma, the maximum shrinking ratio for long path $P_i$ is $\frac{1}{2^{M}-\frac{2^{M}-1}{q_{i,M}}}$ and the minimum shrinking ratio for short path $P_j$ is $\frac{1}{2^{M_{j}}}$, So
	$$
	\frac{q_{i,M}/q_{i,0}}{q_{j,M_j}/q_{j,0}}\leq \frac{2^{M_j}}{2^{M}-\frac{2^{M}-1}{q_{i,M}}}\leq\frac{2^{M-1}}{2^{M}-\frac{2^{M}-1}{q_{i,M}}}\leq 1.
	$$
	Hence, \eqref{cla-52-1} holdds and the lemma is proved.
\end{proof}

We are ready to prove the main lemma of this subsection.

\begin{lemma}\label{lem0108-6}
For a set $\mathcal P_m$ of $m$ vertex-disjoint paths spanning $p_m>2k$ vertices, the set $\mathcal P_m^M$ computed by Algorithm \ref{algo1} contains $m$ vertex-disjoint paths, spans $p_m^M$ vertices with $k<p_m^M\leq 2k$, and has cost
$$
w(\mathcal P_m^{M})\leq 2\cdot \frac{p_m^M}{p_m}\cdot w(\mathcal P_m).
$$
\end{lemma}
\begin{proof}
	By the termination criterion of the algorithm, we have $p'\leq 2k$. In each round of the while loop, the number of vertices spanned by a non-trivial path is reduced by at most a half. Combining this with the fact that at the beginning of the last round before termination, $\mathcal P_m^M$ spans more than $2k$ vertices, we have $p'>k$.
	
	Let $j_l^*$ be the index $j^*$ in line \ref{linej} of Algorithm \ref{algo1} in the $l$-th iteration.
	If $P_l$ is non-trivial, then by Property~\ref{obser331-1},
	\begin{align}\label{eq331-2}
		\frac{w^{(j^*_l)}(P_l)}{q^{(j^*_l)}(P_l)}\geq \frac{w^{(0)}(P_l)+w^{(1)}(P_l)}{q^{(0)}(P_l)+q^{(1)}(P_l)}\geq \frac{w(P_l)}{q(P_l)},
	\end{align}
	where $w^{(0)},w^{(1)},q^{(0)},q^{(1)}$ are defined in the paragraph before Fig. \ref{fig330-1}. It follows that
	\begin{align}\label{eq331-1}
		w(P_{l+1})& = w(P_l)- w^{(j^*_l)}(P_l)\leq w(P_l)-\frac{q^{(j^*_l)}(P_l)}{q(P_l)}w(P_l)\nonumber\\
		&=\frac{q(P_{l+1})}{q(P_l)}w(P_l).
	\end{align}
	Iteratively using inequality \eqref{eq331-1},
	\begin{align}\label{eq331-3}
		w(P_M) \leq\frac{q(P_M)}{q(P_0)}w(P).
	\end{align}
	Note that any path $P\in \mathcal P^{0}$ has $q(P_M)=1$ and cost $w(P_M)=0$. Then by \eqref{eq331-3} and making use of Lemma \ref{lemma-427-1} and Lemma \ref{lemma-427-2},
	\begin{align*}\label{eq427-1}
		&w(\mathcal P_m^M) \\
		=&\sum_{P\in \mathcal P^{>0}}w(P_M)\leq \sum_{P\in \mathcal P^{>0}}\frac{q(P_M)}{q(P_0)}w(P)\\
		\leq &2\cdot\frac{\sum_{P\in \mathcal P^{>0}}q(P_M)}{|V(\mathcal P^{>0})|}w(\mathcal P^{>0})\\
		\leq &2\cdot\frac{\sum_{P\in \mathcal P^{>0}}q(P_M)+\sum_{P\in \mathcal P^{0}}q(P_M)}{|V(\mathcal P^{>0})|+|V(\mathcal P^{0})|}(w(\mathcal P^{>0})+w(\mathcal P^{0}))\\
		= &2\cdot\frac{p_m^M}{p_m} w(\mathcal P_m).
	\end{align*}
	The lemma is proved.
\end{proof}
As a consequence, we have the following result.

\begin{coro}\label{coro0808-1}
If $\mathcal A_{PC}$ computes a path set $\mathcal P_m$ spanning more than $2k$ vertices, then we can find a path set $\mathcal P_m^{M}$ spanning $p_m^M\in[k,2k]$ vertices with
$$
w(\mathcal P_m^{M})\leq 16\cdot\frac{opt_{m,k}+\varepsilon}{1-\alpha}.
$$
\end{coro}

\begin{proof}
	Algorithm \ref{algo1} guarantees $p_m^M\in[k,2k]$. Lemma \ref{coro0108-1} implies $w(\mathcal{P}_m)\leq \frac{4p_m}{(1-\alpha)k}(opt_{m,k}+\varepsilon)$. The result follows from Lemma \ref{lem0108-6} and $p_m^M\leq 2k$.
\end{proof}

To sum up, by applying the algorithm $\mathcal A_{PC}$, we obtain a set $\mathcal P_m$ of $m$ vertex-disjoint paths. If $\mathcal P_m$ spans $p_m\in [\alpha k,2k]$ vertices, then $\mathcal P_m$ is an $(\alpha,\frac{8-2\alpha}{1-\alpha}+\varepsilon)$-bicriteria approximate solution. Otherwise, $p_m>2k$, then we can trim $\mathcal P_m$ into a set $\mathcal P_m^M$ of $m$ vertex-disjoint paths spanning $p_m^M\in[k,2k]$ vertices, which is a feasible solution with approximation ratio at most $\frac{16}{(1-\alpha)}+\varepsilon$. Therefore, we have the following theorem.

\begin{theorem}\label{theo53-1}
For $k$-MinWP$_m$, there exists a polynomial time algorithm which achieves either a $(\frac{16}{(1-\alpha)}+\varepsilon)$-approximation, or an $(\alpha, \frac{8-2\alpha}{1-\alpha}+\varepsilon)$-bicriteria approximation, where $\alpha\in(0,1)$ is an adjustable parameter.
\end{theorem}

Regarding the running time, determining an $L$ to approximate $opt_{m,k}$ and the corresponding path set $\mathcal P_m$ takes time $O(T_{PC}\cdot\frac{\log (k-m)}{\varepsilon})$. The remaining operations can be completed in $O(n)$ time. Therefore, according to Theorem \ref{theo610-1}, the total time required to solve $k$-MinWP$_m$ is $O(\frac{1}{\varepsilon}n^4\log n)$. We denote this algorithm as $\mathcal{A}_{P}^{m,k}$ and its running time as $T_{P}^{m,k}$.

\subsection{A $(0.035-O(\varepsilon))$-approximation for MOP}\label{subsec-002}
In this section, we show how to make use of algorithm $\mathcal A_{P}^{m,k}$ for $k$-MinWP$_m$ to obtain an approximate solution for MOP. The algorithm (denoted as $\mathcal{A}_{MO}^{B,m}$) is presented in Algorithm \ref{algo2}. Parameter $k$ in the outer for-loop is used to guess the optimal value $opt_{MO}$ of MOP. For each guess $k$, compute a path set $\mathcal P_{m,k}$ using algorithm $\mathcal A_{P}^{m,k}$. Consider each path $P\in\mathcal P_{m,k}$ as a line segment of length $w(P)$. Divide this line segment into segments with equal lengths (the number of segments depends on the number of vertices spanned by $\mathcal P_{m,k}$, refer to the ``if'' part of Algorithm \ref{algo2}). Note that some of the line segments might not start and end at vertices. Consider those paths contained in these line segments, and select the one, designated as $P'$, which incorporates the largest number of vertices. We refer to $P'$ as the {\em heaviest sub-path} contained in the line segments. The selected paths $\{P'\colon P\in\mathcal P_{m,k}\}$ form the output of Algorithm \ref{algo2}.

\vskip 0.2cm
\begin{algorithm}
\caption {\small A $(0.035-O(\varepsilon))$-approximation algorithm (denoted as $\mathcal{A}_{MO}^{B,m}$) for MOP.}
\hspace*{0.02in}\raggedright{\bf Input:} An edge-weighted graph $G$, a positive integer $m$ and a budget $B$.

\hspace*{0.02in}\raggedright{\bf Output:} A set $\mathcal P^{B}$ of $m$ vertex-disjoint paths.

\begin{algorithmic}[1]
	
	\FOR{$k=1,\ldots,n$}
	\STATE compute a set $\mathcal{P}_{m,k}$ of $m$ vertex-disjoint paths by calling $\mathcal{A}_{P}^{m,k}$
	\IF{$|V(\mathcal{P}_{k,m})|\geq k$}
	\STATE $ls\leftarrow \lceil\frac{16}{1-\alpha}+\varepsilon\rceil$
	\ELSE 
	\STATE $ls\leftarrow \lceil\frac{8-4\alpha}{1-\alpha}+\varepsilon\rceil$
	\ENDIF
	\FOR{each $P\in\mathcal P_{m,k}$}
	\STATE divide the line segment corresponding to $P$ into $ls$ segments with length $w(P)/ls$
	\STATE $P'\leftarrow$ the heaviest sub-path contained in these line segments
	\ENDFOR
	\STATE $\mathcal{P}'_{m,k}\leftarrow\{P'\colon P\in\mathcal P_{m,k}\}$
	\ENDFOR
	\STATE $k^*\leftarrow\arg\max_{k=1,\ldots,n}\{|V(\mathcal P_{m,k}')|\colon w(\mathcal P_{m,k}')\leq B\}$
	\RETURN $\mathcal P^B\leftarrow \mathcal P'_{m,k^*}$
\end{algorithmic}\label{algo2}
\end{algorithm}

\begin{theorem}\label{the624-1}
Algorithm \ref{algo2} provides a $(0.035-O(\varepsilon))$-approximation for MOP and runs in $O(\frac{1}{\varepsilon}n^4\log n)$ time.
\end{theorem}

\begin{proof}
	Consider the case when $k=opt_{MO}$ in Algorithm \ref{algo2}, i.e, the guessed value matches the optimal value, by Theorem \ref{theo53-1},  $\mathcal{A}_{P}^{m,opt_{MO}}$ either computes $m$ disjoint paths spanning at least $opt_{MO}$ vertices with cost at most $(\frac{16}{1-\alpha}+\varepsilon)B$, or computes $m$ disjoint paths spanning at least $\alpha\cdot opt_{MO}$ vertices with cost at most $(\frac{8-4\alpha}{1-\alpha}+\varepsilon)B$.
	If the former case occurs, then $\mathcal{P}'_{m,opt_{MO}}$ contains at least $$\frac{opt_{MO}}{\lceil\frac{16}{1-\alpha}+\varepsilon\rceil}\geq
	\frac{opt_{MO}}{\frac{17-\alpha}{1-\alpha}+\varepsilon}=
	\left(\frac{1-\alpha}{17-\alpha}-O(\varepsilon)\right)\cdot opt_{MO}$$
	vertices and
	$$
	w(\mathcal{P}'_{m,opt_{MO}})\leq \frac{w(\mathcal P_{m,opt_{MO}})}{\lceil\frac{16}{1-\alpha}+\varepsilon\rceil}\leq B.
	$$
	If the latter case occurs, then $\mathcal{P}'_{m,opt_{MO}}$ contains at least $$\frac{\alpha\cdot opt_{MO}}{\lceil\frac{8-4\alpha}{1-\alpha}+\varepsilon\rceil}\geq\frac{\alpha\cdot opt_{MO}}{\frac{9-5\alpha}{1-\alpha}+\varepsilon}=
	\left(\frac{\alpha(1-\alpha)}{9-5\alpha}-O(\varepsilon)\right)\cdot opt_{MO}$$
	vertices and
	$$
	w(\mathcal{P}'_{opt_{MO},m})\leq \frac{w(\mathcal P_{m,opt_{MO}})}{\lceil\frac{8-4\alpha}{1-\alpha}+\varepsilon\rceil}\leq B.
	$$
	Therefore, $\mathcal{P}'_{opt_{MO},m}$ is a feasible solution to MOP with approximation ratio at least
	$$
	h(\alpha)=\min\left\{f(\alpha)=\frac{1-\alpha}{17-\alpha},  \ g(\alpha)=\frac{\alpha(1-\alpha)}{9-5\alpha}\right\}-O(\varepsilon),
	$$
	where $\alpha\in(0,1)$. It can be verified that the above minimum achieves its maximum value at $\alpha^*=11-4\sqrt{7}$, where $f(\alpha^*)=g(\alpha^*)$. Hence,
	$$
	h(\alpha^*)=\frac{4\sqrt{7}-10}{6+4\sqrt{7}}-O(\varepsilon)\approx 0.035-O(\varepsilon).
	$$
	
	Regarding running time, the most time-intensive aspect is calling the $\mathcal A_{P}^{m,k}$ algorithm $n$ times, which requires $O(nT_{P}^{k,m})$ time. The remaining operations can be completed in $O(n)$ time. Combining this with the previous analysis of $T_P^{k,m}$, the total running time is $O(\frac{1}{\varepsilon}n^5\log n)$.
\end{proof}

\section{Approximation algorithm for BSC}
In this section, a $(0.0116-O(\varepsilon))$-approximation algorithm is presented for BSC based on algorithm $\mathcal{A}_{MO}^{B,m}$ for MOP. The algorithm consists of two steps. The first step is {\em vertex grouping}, which computes a set $\mathcal P_N$ of $N$ vertex-disjoint paths by calling Algorithm~\ref{algo2}. The second step is {\em sensor allocation}. Note that for the {\em BSC on a line}, which asks for the maximum number of PoIs on a line $L$ to be sweep-covered by $N$ mobile sensors with the same speed, \cite{Diyan2021} has given an $O(|V(L)|N)$-time algorithm to compute an optimal solution, where $|V(L)|$ is the number of PoIs on $L$ (denote this algorithm as $\mathcal A_{line}^{L,N}$, and denote
the optimal value as $opt_{line}^{L,N}$). So, the key to the allocation step is to determine the number of mobile sensors $N_i$ to be deployed on the $i$-th path $P_i\in\mathcal P_N$. To achieve this objective, we employ a dynamic programming approach as outlined below. Suppose $\mathcal P_N=\{P_1,\ldots,P_N\}$. For $i=1,\ldots,N$, denote by $\mathcal P_i=\{P_1,\ldots,P_i\}$ and let $\mathcal P_0=\emptyset$. For $K_i\in\{0,1,\ldots,N\}$ and $N_i\in\{0,1,\ldots,K_i\}$, let $c(i,K_i,N_i)$ be the maximum number of vertices in $\mathcal P_i$ that can be sweep covered by $K_i$ mobile sensors such that path $P_i$ is allocated exactly $N_i$ mobile sensors.
The following lemma shows how to compute the values $\{c(i,K_i,N_i)\}$.

\begin{lemma}\label{lemma-629-1}
The values $\{c(i,K_i,N_i)\}$ can be computed using the following formula
$$
c(i,K_i,N_i)
=\max\{c(i-1,K_i-N_i,j)\}
+opt_{line}^{P_i,N_i}
$$
where the maximum is taken over all
$$j\in\{0,1,\ldots,\min\{K_i-N_i,|V(P_{i-1})|\}\}.$$
The initial conditions are $c(0,K_i,N_i)=0$ for any $K_i, N_i$ and $c(i,0,0)=0$ for any $i$.
\end{lemma}

Let $\{N_i^*\}_{i=1}^N$ be the optimal assignment of mobile sensors on the path set $\mathcal P_N$. After computing all $\{c(i,K_i,N_i)\}$ values using Lemma \ref{lemma-629-1}, we can determine the optimal assignments $\{N_i^*\}_{i=1}^N$ as follows:
\begin{equation}\label{recusive1}
\displaystyle N_N^* =\arg\max\limits_{N_N\in\{0,1,\ldots,N\}}\ c(N,N,N_N)
\end{equation}
(that is, the optimal number of mobile sensors assigned to the $N$-th path achieves the maximum value among $\{c(N,N,N_N)\}_{N_N=0,1,\ldots,N}$)
and for $i=N-1,N-2,\ldots,1$,
\begin{equation}\label{recusive2}
\displaystyle N_i^* =\arg\max
c(i,N-\sum_{j=i+1}^NN_j^*,N_i),
\end{equation}
where the maximum is taken over all $N_i\in\{0,1,\ldots,$ $N-\sum_{j=i+1}^NN_j^*\}$ (that is, having determined $N_{i+1}^*,\ldots,N_N^*$, the total number of sensors assigned to the first $i$ paths is $N-\sum_{j=i+1}^NN_j^*$, and assigning $N_i^*$ sensors to the $i$th path will achieve the maximum value among $\{c(i,N-\sum_{j=i+1}^NN_j^*,N_i)\}_{N_i\in\{0,1,\ldots,N-\sum_{j=i+1}^NN_j^*\}}$). The running time of determining $\{N_i^*\}_{i=1}^N$ is at most $O(N^4)$.

The complete algorithm is detailed  in Algorithm \ref{algo3}.
\begin{algorithm}
\caption {a constant-approximation algorithm for BSC}
\hspace*{0.02in}\raggedright{\bf Input:} A metric graph $G$ and a positive integer $N$.

\hspace*{0.02in}\raggedright{\bf Output:} A schedule of routes for $N$ mobile sensors.

\begin{algorithmic}[1]
	\STATE Compute $N$ vertex-disjoint paths $\mathcal{P}_N$ by algorithm $\mathcal{A}_{MO}^{B,m}$ with $B=Nat, m=N$;\label{line0726-1}
	\STATE determine $\{N_i^*\}_{i=1}^N$ using recursive formula \eqref{recusive1} \eqref{recusive2}; \label{line0726-2}
	\FOR {each $P\in\mathcal{P}_N$}
	\STATE compute an optimal route using algorithm $\mathcal A_{line}^{P_i,N_i^*}$;
	\ENDFOR
	\RETURN the union of the above routes.
\end{algorithmic}\label{algo3}
\end{algorithm}

\begin{theorem}\label{lemma-629-2}
Algorithm \ref{algo3} is a $(0.0116-O(\varepsilon))$-approximation for BSC and executes in time $O(\frac{1}{\varepsilon}n^4\log n)$, where $n$ is the number of vertices.
\end{theorem}

\begin{proof}
	Let $opt_{BSC}$ be the optimal value of the BSC instance. During time interval $[0,t]$, each mobile sensor $s_i$ in an optimal solution travels a walk $W_i^*$ of length $at$. We can trim these walks into a set $\mathcal P_N^*$ of at most $N$ vertex-disjoint paths as follows. Walk along $W_1^*$ and short-cut repeated vertices to get a path $P_1^*$. Next, walk along $W_2$ to obtain a path $P_2^*$ by short-cutting both repeated vertices of $W_2^*$ and those vertices on $P_1^*$. Proceeding in this way until all walks are processed. Note that $\mathcal P_N^*$ spans all those vertices sweep-covered by the optimal solution and thus $|V(\mathcal P_N^*)|=opt_{BSC}$. Furthermore, $w(\mathcal P_N^*)\leq \sum_{i=1}^N w(W_i^*)=Nat$. So, $\mathcal P_N^*$ is a feasible solution to the MOP instance with budget $Nat$ that spans $opt_{BSC}$ vertices. Because the path set $\mathcal{P}_N$ computed by $\mathcal{A}_{MO}^{Nat,N}$ is a $(0.035-O(\varepsilon))$-approximate solution to the same MOP instance, we have
	\begin{equation}\label{eq0724-1}
		w(\mathcal P_N)\leq Nat
	\end{equation}
	and
	\begin{eqnarray}\label{eq0703-1}
		|V(\mathcal{P}_N)|  \geq \left(0.035-O(\varepsilon)\right)|V(\mathcal{P}_N^*)|
		=\left(0.035-O(\varepsilon)\right) opt_{BSC}. \nonumber
	\end{eqnarray}
	
	\textbf{\em Claim.} There exists a sweep coverage scheme for $N$ mobile sensors to sweep-cover at least $|V(\mathcal{P}_N)|/3$ vertices of $\mathcal P_N$.
	
	Because a mobile sensor with speed $a$ can travel a distance of $at$ in a time span $t$, if we let it travel back and forth along a line segment of length $at/2$, then all points on this line segment are sweep-covered within the time span $t$. Consider a line segment of length $at/2$ as a {\em block}. By placing these blocks side by side and assigning one mobile sensor to sweep-cover each block, we call this scheme the {\em separate working mode}. Using separate working mode, all vertices on a path $P$ of length $w(P)$ can be sweep-covered by $\lceil w(P_i)/{\frac{at}{2}}\rceil\leq \frac{2w(P_i)}{at}+1$ mobile sensors, and thus the number of mobile sensors needed to sweep-cover all vertices of $\mathcal P_N$ is at most
	\begin{align}\label{eq0703-2}
		\sum_{P_i\in \mathcal{P}_N}\left(\frac{2w(P_i)}{at}+1\right)  &\leq
		\frac{2\sum_{P_i\in \mathcal{P}_N}w(P_i)}{at}+N\\
		&=\frac{2w(\mathcal{P}_N)}{at}+N\leq 3N,\nonumber
	\end{align}
	where the last inequality makes use of \eqref{eq0724-1}. In other words, at most $3N$ blocks are sufficient to cover all these paths. Choosing $N$ {\em richest} blocks from them, where the richness of a block is measured by the number of vertices contained in it, we can effectively sweep-cover no less than $|V(\mathcal{P})|/3$ vertices using separate mode. The claim is proved.
	
	Combining the claim with inequality \eqref{eq0703-1}, there is a sweep coverage scheme of $N$ mobile sensors sweep-covering at least $\frac13\left(0.035-O(\varepsilon)\right) opt_{BSC}\geq\left(0.0116-O(\varepsilon)\right)$ $opt_{BSC}$ vertices.
	As for the running time, calling the algorithm for MOP (line \ref{line0726-1}) takes time $O(\frac{1}{\varepsilon}n^4\log n)$ (see Theorem \ref{the624-1}); determining $\{N_i^*\}_{i=1}^N$ (line \ref{line0726-2}) takes time $O(N^4)=O(n^4)$ (see the comment below Lemma \ref{lemma-629-1}); the for loop takes time $O(nN)=O(n^2)$.  Therefore, Algorithm~\ref{algo3} executes in time $\frac{1}{\varepsilon}n^4\log n$ and the theorem is proved.
\end{proof}

\section{Conclusion and Future Work}
This paper introduces the first constant approximation algorithm for BSC, with an approximation ratio of $(0.0116-O(\varepsilon))$. We present a novel approach to assign sensors for multiple paths. We use dynamic programming in computing an optimal sensor allocation. Although, there are complexities in computing an optimum sweep-route, surprisingly, an effective approximation of an optimum can be obtained by simply allowing the mobile sensors to act independently. One interesting direction that deserves further study is the optimal allocation of sensors especially under different speeds and sweep periods.



\ifCLASSOPTIONcaptionsoff
\newpage
\fi

\bibliographystyle{IEEEtran}
\bibliography{BSC-ref}



%

%

\begin{IEEEbiography}[{\includegraphics[width=1in,height=1.25in,clip,keepaspectratio]{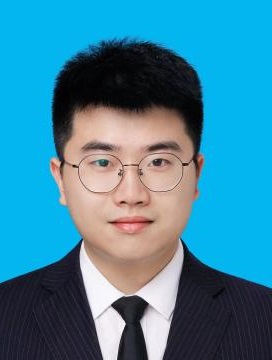}}]{Wei Liang}
received his bachelor degree from Jiangxi Science and Technology Normal University in 2019. He entered Zhejiang Normal University for his master degree the same year. He is now a doctoral student at Zhejiang Normal University. His main interest is about approximation algorithms for coverage problems in networks.
\end{IEEEbiography}
\vskip -1cm
\begin{IEEEbiography}[{\includegraphics[width=1in,height=1.25in,clip,keepaspectratio]{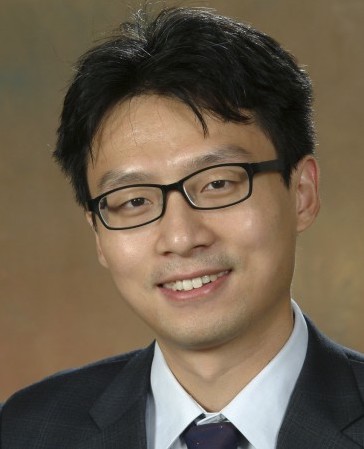}}]{Shaojie Tang}
is currently an assistant professor of Naveen Jindal School of Management at University of Texas at Dallas. He received the PhD degree in computer science from Illinois Institute of Technology, in 2012. His research interests include social networks, mobile commerce, game theory, e-business, and optimization.
He received the Best Paper Awards in ACM MobiHoc 2014 and IEEE MASS 2013. He also received the ACM SIGMobile service award in 2014. Dr. Tang served in various positions (as chairs and TPC members) at numerous conferences, including ACM MobiHoc and IEEE ICNP. He is an editor for International Journal of Distributed Sensor Networks.
\end{IEEEbiography}

\vskip -1cm
\begin{IEEEbiography}[{\includegraphics[width=1in,height=1.25in,clip,keepaspectratio]{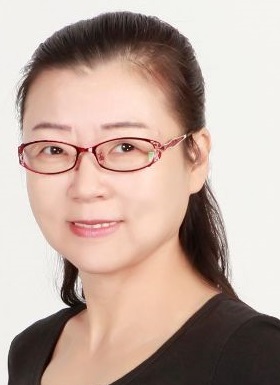}}]{Zhao Zhang}
received her PhD degree from Xinjiang University in 2003. She worked in Xinjiang University from 1999 to 2014, and now is a professor in School of Mathematical Sciences, Zhejiang Normal University. Her main interest is in combinatorial optimization, especially approximation algorithms for NP-hard
problems which have their background in networks.
\end{IEEEbiography}





\end{document}